\documentclass[english]{article}
\usepackage[T1]{fontenc}
\usepackage[latin9]{inputenc}
\usepackage{xcolor}
\usepackage{amsmath}
\usepackage{amsthm}
\usepackage{amssymb}
\PassOptionsToPackage{normalem}{ulem}
\usepackage{ulem}

\makeatletter

\providecolor{lyxadded}{rgb}{0,0,1}
\providecolor{lyxdeleted}{rgb}{1,0,0}

\DeclareRobustCommand{\lyxsout}[1]{\ifx\\#1\else\sout{#1}\fi}

\newcommand{\lyxaddress}[1]{
	\par {\raggedright #1
	\vspace{1.4em}
	\noindent\par}
}
\theoremstyle{plain}
\newtheorem{thm}{\protect\theoremname}[section]
\theoremstyle{definition}
\newtheorem{defn}[thm]{\protect\definitionname}
\ifx\proof\undefined
\newenvironment{proof}[1][\protect\proofname]{\par
	\normalfont\topsep6\p@\@plus6\p@\relax
	\trivlist
	\itemindent\parindent
	\item[\hskip\labelsep\scshape #1]\ignorespaces
}{%
	\endtrivlist\@endpefalse
}
\providecommand{\proofname}{Proof}
\fi

\@ifundefined{date}{}{\date{}}
\usepackage{geometry}
\usepackage{cite}

\newcommand{\T}{\mathsf{T}}
\newcommand{\C}{\mathsf{C}}
\newcommand{\X}{\mathsf{X}}

\makeatother

\usepackage{babel}
\providecommand{\definitionname}{Definition}
\providecommand{\theoremname}{Theorem}

\begin{document}
\title{Contextuality with disturbance and without: Neither can violate substantive
requirements the other satisfies}
\author{Ehtibar N. Dzhafarov\textsuperscript{1} and Janne V. Kujala\textsuperscript{2}}
\maketitle

\lyxaddress{\begin{center}
\textsuperscript{1}Purdue University, USA, ehtibar@purdue.edu\\
 \textsuperscript{2}University of Turku, Finland, jvk@iki.fi
\par\end{center}}
\begin{abstract}
Contextuality was originally defined only for consistently connected
systems of random variables (those without disturbance/signaling).
Contextuality-by-Default theory (CbD) offers an extension of the notion
of contextuality to inconsistently connected systems (those with disturbance),
by defining it in terms of the systems' couplings subject to certain
constraints. Such extensions are sometimes met with skepticism. We
pose the question of whether it is possible to develop a set of substantive
requirements (i.e., those addressing a notion itself rather than its
presentation form) such that (1) for any consistently connected system
these requirements are satisfied, but (2) they are violated for some
inconsistently connected systems. We show that no such set of requirements
is possible, not only for CbD but for all possible CbD-like extensions
of contextuality. This follows from the fact that any extended contextuality
theory $\T$ is contextually equivalent to a theory $\T'$ in which
all systems are consistently connected. The contextual equivalence
means the following: there is a bijective correspondence between the
systems in $\T$ and $\T'$ such that the corresponding systems in
$\T$ and $\T'$ are, in a well-defined sense, mere reformulations
of each other, and they are contextual or noncontextual together.

\end{abstract}
A formal theory $\T$ of contextuality is defined by a class $\mathfrak{R}$
of possible systems of random variables and a rule by which these
systems are divided into noncontextual and contextual ones. In the
original theory of contextuality ( in which term we include both the
Kochen-Specker contextuality and the contextuality in distributed
systems, referred to as nonlocality \cite{AbramskyBrandenburger2011,Bell1964,Bell1966,Cabello2008,KurzynskiRamanathanKaszlikowski(2012),SuppesZanotti1981,KochenSpecker1967,Review2022}),
the class $\mathfrak{R}$ is confined to consistently connected systems,
or a subclass thereof. These are the systems with no ``disturbance''
or ``signaling,'' which means that the variables representing the
same property (answering the same question) in different contexts
are identically distributed. The Contextuality-by-Default theory (CbD)
extends the notion of contextuality to all systems of random variables,
including those with disturbance \cite{DzhCerKuj2017,KD2021}, and
it has been applied to several experimental and theoretical situations
\cite{Fluhmanetal.2018,Zhanetal2017,Bacciagaluppi(2015),Amaral etal.2018,Khrennikov2022,Kupzc2021,Malinowski2018,Wangetal2022}.
A recent workshop on contextuality \cite{workshop} exhibited a renewed
interest to studying contextuality in inconsistently connected systems,
including approaches that are distinctly non-CbD-like \cite{SAB.QCQMB22,ABS.QCQMB22,Mansfield.QCQMB22},
and some work directly critical of CbD (\cite{Jones.QCQMB22}, responded
to in Ref. \cite{DzhKuj2023Nonmeasurements}).

The present paper is not about CbD specifically. Rather it is about
a broad class of all possible \emph{CbD-like theories}, as defined
below. The plan and the main message of the paper are as follows.
In Section 1, we present the terminology and notation to be used,
and define the notion of a system of random variables modeling (representing,
describing) another system. In Section 2, we define the traditional
notion of contextuality in the language of probabilistic couplings
\cite{Thorisson(2000)}, and we introduce the notion of $\C$-contextuality
as a very broad generalization of both traditional contextuality and
CbD-contextuality. In Section 3, we introduce the notion of consistification
of a system, and show that any theory $\T$, irrespective of its class
$\mathfrak{R}$ of systems and a specific version of the $\C$-contextuality
it uses, can be redefined as a theory $\T'$, whose systems are consistently
connected, and that uses the traditional notion of contextuality.
Because of this, we conclude, there can be no set of substantive requirements
$\X$ for the notion of contextuality that are satisfied by all consistently
connected systems but contravened by some inconsistently connected
ones. Indeed, if such a set of requirements existed, one could form
a theory $\T$ whose class $\mathfrak{R}$ includes some systems contravening
$\X$. But $\X$ would then be satisfied by the theory $\T'$ that
is contextually equivalent to $\T$ and a mere reformulation thereof.
Consequently, requirements $\X$ cannot be substantive: they address
a form rather than the substance of the notion of contextuality. In
Section 4, we discuss some issues related to the consistified systems
(the term used for the consistently connected systems in $\T'$),
including the representability thereof by hidden variable models.
We also briefly discuss there a still more general (in fact, maximally
general) notion of $\C$-contextuality, one that does not have the
existence-and-uniqueness property postulated for $\C$-contextuality.
In the final analysis this does not alter the main conclusion of the
paper.

The idea that consistification precludes the possibility of rejecting
extended contextuality while accepting the traditional one was previously
mentioned in Ref. \cite{DzhKuj2023Nonmeasurements}. However, it was
confined to CbD only, and mentioned without elaborating. The consistification
procedure was first described in Ref. \cite{Amaral etal.2018} for
an older version of the CbD approach, and it was elaborated and adapted
to the current version of CbD in Ref. \cite{chapter}. Finally, the
$\C$-contextuality in our paper generalizes a more limited version
of $\C$-contextuality that was used in Ref. \cite{Dzh2017Nothing}
as a generalization of the CbD approach.

\section{Basic notions}

A \emph{system} of random variables is a set of double-indexed random
variables
\begin{equation}
\mathcal{R}=\left\{ R_{q}^{c}:c\in C,q\in Q,q\prec c\right\} ,\label{eq:system}
\end{equation}
where $q\in Q$ identifies what the random variable $R_{q}^{c}$ represents
(measures, responds to, describes), $c\in C$ identifies circumstances
under which $R_{q}^{c}$ is recorded (including what other random
variables are recorded together with $R_{q}^{c}$). $q$ and $c$
are referred to as, respectively, the \emph{content} and the \emph{context}
of the random variable $R_{q}^{c}$. The relation $q\prec c$ indicates
that a variable with content $q$ is recorded in context $c$. As
an example, this is a system with $Q=\left\{ 1,2,3\right\} $ and
$C=\left\{ 1,2,3,4\right\} $: 
\begin{equation}
\begin{array}{|c|c|c||c|}
\hline R_{1}^{1} & R_{2}^{1} &  & c=1\\
\hline  & R_{2}^{2} & R_{3}^{2} & c=2\\
\hline R_{1}^{3} &  & R_{3}^{3} & c=3\\
\hline R_{1}^{4} & R_{2}^{4} & R_{3}^{4} & c=4\\
\hline\hline q=1 & q=2 & q=3 & \mathcal{R}
\\\hline \end{array}\:.\label{eq:example system}
\end{equation}

The subset $R^{c}=\left\{ R_{q}^{c}:q\in Q,q\prec c\right\} $ of
random variables recorded in the same context $c$ (a row in the matrix
above) is termed a \emph{bunch}, and the subset $\mathcal{R}_{q}=\left\{ R_{q}^{c}:c\in C,q\prec c\right\} $
of random variables sharing a content $q$ (a column in the matrix
above) is termed a \emph{connection}. The difference in font ($R^{c}$
vs $\mathcal{R}_{q}$) reflects the fact that $R^{c}$ is a random
variable in its own right (i.e., all its components are jointly distributed),
whereas the components of $\mathcal{R}_{q}$ are not jointly distributed.
In fact, no two random variables $R_{q}^{c}$ and $R_{q'}^{c'}$ are
jointly distributed unless they are in the same bunch, $c=c'$. The
measurable space on which $R_{q}^{c}$ is distributed is assumed to
be the same for all elements of a connection, and can be denoted $\left(A_{q},\Sigma_{q}\right)$.

The triple $\left(Q,C,\prec\right)$ is called the \emph{format} of
the system. It is essentially the mathematical depiction of ``what
the system is about,'' what kind of empirical or theoretical situation
it represents. Thus, the format of the system in $\left(\ref{eq:example system}\right)$
can be presented as 
\begin{equation}
\begin{array}{|c|c|c||c|}
\hline \star & \star &  & c=1\\
\hline  & \star & \star & c=2\\
\hline \star &  & \star & c=3\\
\hline \star & \star & \star & c=4\\
\hline\hline q=1 & q=2 & q=3 & \mathcal{R}\textnormal{'s format}
\\\hline \end{array}\:,
\end{equation}
where $\star$ indicates the elements of the relation $\prec$. To
define a system of a given format, one has to specify the distributions
of its bunches

As should be clear from the abstract and introduction, in this paper
we use the notion of one system of random variables, $\mathcal{B}$,
being a ``mere reformulation'' of another, $\mathcal{A}$. Intuitively,
this means that whatever empirical or theoretical situation is modeled
(described, represented) by $\mathcal{A}$, it is also modeled by
$\mathcal{B}$. The relation between a system and a situation it depicts
is difficult to formalize directly, as one would have then to impose
some formal structure on the situation being represented before it
is represented (as it is done in the representational theory of measurement,
\cite{Pfanzagl1968,Narens1985}). However, it is sufficient for our
purposes to formalize a simpler relationship: between a system $\mathcal{A}$
and another system that models (describes, represents) the system
$\mathcal{A}$. Moreover, rather than presenting this relationship
in a most general possible way, it will suffice to describe one special,
universally applicable construction of the modeling systems $\mathcal{B}.$
We will refer to this construction as \emph{canonical modeling}.

Consider two classes of systems, $\mathfrak{R}$ and $\mathfrak{R}^{\dagger}$,
in a bijective correspondence to each other, about which we say that
any system in $\mathfrak{R}$ is canonically modeled by the corresponding
system in $\mathfrak{\mathfrak{R}^{\dagger}}$. The following definition
gives a precise meaning to this relation.
\begin{defn}
\label{def:modeling}We say that a system $\mathcal{R}\in\mathfrak{R}$
with format $\left(Q,C,\prec\right)$ is canonically modeled by a
system $\mathcal{R}^{\dagger}\in\mathfrak{\mathfrak{R}^{\dagger}}$
with format $\left(Q^{\dagger},C^{\dagger},\prec^{\dagger}\right)$
if
\begin{description}
\item [{(canonical$\:$contents)}] $Q^{\dagger}=\left\{ \left(q,c\right):q\prec c\right\} $,
\item [{(canonical$\:$contexts)}] $C^{\dagger}=\left\{ \left(\cdot,c\right):c\in C\right\} \sqcup\left\{ \left(q,\cdot\right):q\in Q\right\} $,
\item [{(canonical$\:$relation)}] $\left(q,c\right)\prec^{\dagger}\left(\cdot,c\right)\Longleftrightarrow\left(q,c\right)\in Q^{\dagger}$,
and $\left(q,c\right)\prec^{\dagger}\left(q,\cdot\right)\Longleftrightarrow\left(q,c\right)\in Q^{\dagger},$
\item [{(main$\:$bunches)}] $R^{\left(\cdot,c\right)}=\left\{ R_{\left(q,c\right)}^{\left(\cdot,c\right)}:\left(q,c\right)\prec^{\dagger}\left(\cdot,c\right)\right\} \overset{d}{=}\left\{ R_{q}^{c}:q\prec c\right\} =R^{c}$
\item [{(auxiliary$\:$bunches)}] $R^{\left(q,\cdot\right)}=\left\{ R_{\left(q,c\right)}^{\left(q,\cdot\right)}:\left(q,c\right)\prec^{\dagger}\left(q,\cdot\right)\right\} $
is uniquely determined by the distributions of the corresponding variables
in $\mathcal{R}_{\left(q,\cdot\right)}=\left\{ R_{\left(q,c\right)}^{\left(\cdot,c\right)}:\left(q,c\right)\prec^{\dagger}\left(q,\cdot\right)\right\} $.
\end{description}
\end{defn}
Here, the symbol $\overset{d}{=}$ stands for ``has the same distribution
as.'' The dot symbol in $\left(\cdot,c\right)$ and $\left(q,\cdot\right)$
should be taken as part of the names of these contexts. We choose
this notation to emphasize that every random variable $R_{q}^{c}$
of the system $\mathcal{R}$ is placed in $\mathcal{R}^{\dagger}$
within two contexts, $\left(\cdot,c\right)$ and $\left(q,\cdot\right)$,
whose names are derived from the indices of the variable. Note that
the variables in the set $\mathcal{R}_{\left(q,\cdot\right)}$ defined
here have the same distributions as the corresponding variables in
$\mathcal{R}_{q}=\left\{ R_{q}^{c}:q\prec c\right\} $. We use the
former set, however, to emphasize that the auxiliary bunches are uniquely
determined by the corresponding variables in the main bunches. Note
that the variables in $\mathcal{R}_{\left(q,\cdot\right)}$ are not
jointly distributed, so $R^{\left(q,\cdot\right)}$ depends on their
individual distributions only.

To give an example, consider the systems 
\begin{equation}
\begin{array}{|c|c||c|}
\hline R_{1}^{1} & R_{2}^{1} & c=1\\
\hline R_{1}^{2} & R_{2}^{2} & c=2\\
\hline\hline q=1 & q=2 & \mathcal{A}
\\\hline \end{array}\quad\begin{array}{|c|c|c|c||c|}
\hline R_{\left(1,1\right)}^{\left(\cdot,1\right)} & R_{\left(2,1\right)}^{\left(\cdot,1\right)} &  &  & c=\left(\cdot,1\right)\\
\hline  &  & R_{\left(1,2\right)}^{\left(\cdot,2\right)} & R_{\left(2,2\right)}^{\left(\cdot,2\right)} & c=\left(\cdot,2\right)\\
\hline R_{\left(1,1\right)}^{\left(1,\cdot\right)} &  & R_{\left(1,2\right)}^{\left(1,\cdot\right)} &  & c=\left(1,\cdot\right)\\
\hline  & R_{\left(2,1\right)}^{\left(2,\cdot\right)} &  & R_{\left(2,2\right)}^{\left(2,\cdot\right)} & c=\left(2,\cdot\right)\\
\hline\hline q=\left(1,1\right) & q=\left(2,1\right) & q=\left(1,2\right) & q=\left(2,2\right) & \mathcal{B}
\\\hline \end{array}\:.\label{eq:QOE}
\end{equation}
Observe that in system $\mathcal{B}$ the contents, contexts, and
the relation between them are constructed in accordance with Definition
\ref{def:modeling}. System $\mathcal{B}$ canonically models system
$\mathcal{A}$ if
\begin{equation}
\left\{ R_{\left(1,1\right)}^{\left(\cdot,1\right)},R_{\left(2,1\right)}^{\left(\cdot,1\right)}\right\} \overset{d}{=}\left\{ R_{1}^{1},R_{2}^{1}\right\} ,\qquad\left\{ R_{\left(1,2\right)}^{\left(\cdot,2\right)},R_{\left(2,2\right)}^{\left(\cdot,2\right)}\right\} \overset{d}{=}\left\{ R_{1}^{2},R_{2}^{2}\right\} ,
\end{equation}
and if there is a rule by which the distribution of 
\begin{equation}
R^{\left(q,\cdot\right)}=\left\{ R_{\left(q,1\right)}^{\left(q,\cdot\right)},R_{\left(q,2\right)}^{\left(q,\cdot\right)}\right\} ,\qquad q=1,2,
\end{equation}
is uniquely determined by the distributions of the corresponding variables
in
\begin{equation}
\mathcal{R}_{\left(q,\cdot\right)}=\left\{ R_{\left(q,1\right)}^{\left(\cdot,1\right)},R_{\left(q,2\right)}^{\left(\cdot,2\right)}\right\} ,\qquad q=1,2.
\end{equation}

Observe the following properties of canonical modeling.
\begin{enumerate}
\item The formats of $\mathcal{R}$ and $\mathcal{R}^{\dagger}$ are reconstructible
from each other, and so are the bunches of the two systems. Moreover,
$\mathcal{R}^{\dagger}$ faithfully replicates the bunches of $\mathcal{R}$.
This allows one to say that $\mathcal{R}$ and $\mathcal{R}^{\dagger}$
describe the same empirical or theoretical situation.
\item One might wonder why we need the auxiliary contexts at all, and they
are indeed unnecessary if all one wants is a system modeling another
system, e.g.,
\[
\begin{array}{|c|c|c|c||c|}
\hline R_{\left(1,1\right)}^{\left(\cdot,1\right)} & R_{\left(2,1\right)}^{\left(\cdot,1\right)} &  &  & c=\left(\cdot,1\right)\\
\hline  &  & R_{\left(1,2\right)}^{\left(\cdot,2\right)} & R_{\left(2,2\right)}^{\left(\cdot,2\right)} & c=\left(\cdot,2\right)\\
\hline\hline q=\left(1,1\right) & q=\left(2,1\right) & q=\left(1,2\right) & q=\left(2,2\right) & \mathcal{B}'
\\\hline \end{array}
\]
 However, we will see the utility of the auxiliary contexts when we
introduce consistifications and contextual equivalence, in Section
\ref{sec:Impossibility-theorem}.
\item The contents in the modeling system are ``contextualized.'' For
instance, system $\mathcal{A}$ in (\ref{eq:QOE}) may be describing
an experiment in which two questions, $q=1$ and $q=2$, are asked
in two orders, $c=1$ indicating ``1 then 2,'' and $c=2$ indicating
``2 then 1'' \cite{Moore2022,WangBusemeyer2013}. In this case,
in the modeling system the content $q=\left(1,2\right)$ should be
interpreted as ``question 1 asked second,'' and $q=\left(1,1\right)$
as ``question 1 asked first.'' We will return to the issue of interpretation
in Section \ref{subsec:Interpretation-of-contents}.
\item The indexation of the variables in a canonical model is clearly redundant,
and it can be simplified. It is more important, however, to maintain
the general logic of indexing the variables by their contents and
contexts.
\end{enumerate}

\section{Traditional and extended contextuality}

A system $\mathcal{R}$ is \emph{consistently connected} if in every
connection $\mathcal{R}_{q}$ all its constituent variables have one
and the same distribution. Otherwise the system is \emph{inconsistently
connected.} (The latter term is also used to designate arbitrary systems,
i.e., in the meaning of ``not necessarily consistently connected.'')

An \emph{overall coupling} of a system $\mathcal{R}$ in (\ref{eq:system})
is an identically labelled system 
\begin{equation}
S=\left\{ S_{q}^{c}:c\in C,q\in Q,q\prec c\right\} 
\end{equation}
of \emph{jointly distributed} random variables such that its bunches
$S^{c}$ are distributed as the corresponding bunches $R^{c}$, 
\begin{equation}
S^{c}\overset{d}{=}R^{c}.
\end{equation}
Clearly, $S$ has the same format as $\mathcal{R}$. A \emph{coupling}
$S_{q}$ of a connection $\mathcal{R}_{q}$ is a set
\begin{equation}
S_{q}=\left\{ S_{q}^{c}:c\in C,q\prec c\right\} 
\end{equation}
of jointly distributed random variables such that $S_{q}^{c}\overset{d}{=}R_{q}^{c}$
for all its elements. A connection coupling $S_{q}$ is said to be
an \emph{identity coupling }if $S_{q}^{c}=S_{q}^{c'}$ for any two
of its elements. Obviously, such a coupling exists if and only if
all of its elements (equivalently, all elements of the connection
$\mathcal{R}_{q}$) have one and the same distribution. Moreover,
the identity coupling is unique if it exists. (The uniqueness of a
coupling should always be understood as the uniqueness of its distribution.
In other words, it is irrelevant on what domain probability space
the coupling is defined as a random variable.)

The traditional notion of contextuality is confined to consistently
connected systems, and it can be rigorously defined in our terminology
as follows.
\begin{defn}
\label{def:traditional}A consistently connected system $\mathcal{R}\in\mathfrak{R}$
is noncontextual if it has a coupling $S$ in which any connection
$S_{q}$ is the identity coupling of the connection $\mathcal{R}_{q}$.
Otherwise the system is contextual.
\end{defn}
The class of all possible systems $\mathcal{R}$ in a theory $\T$
is denoted $\mathfrak{R}$. For instance, $\mathfrak{R}$ can only
contain the systems with finite sets $Q$ and $C$, or only the systems
with dichotomous random variables. By constraining the class $\mathfrak{R}$
one induces constraints on all possible random variables, $R_{q}^{c}\in\mathcal{\mathfrak{R}}_{+}^{+}$,
on bunches of random variables, $R^{c}\in\mathcal{\mathfrak{R}}^{+}$,
and on possible connections, $\mathcal{R}_{q}\in\mathcal{\mathfrak{R}}_{+}$.

In CbD, contextuality of a system $\mathcal{R}$ is defined by considering
its couplings $S$ and determining if in some of them the couplings
$S_{q}$ of the system's connections $\mathcal{R}_{q}$ satisfy a
certain statement. To generalize this definition to all possible CbD-like
theories, all one has to do is to replace this specific statement
with one that is (almost) arbitrary. Let $\C$ be any statement of
the form ``the coupling of connection $\mathcal{R}_{q}$ has the
following properties: ...''. The only constraints we impose on $\C$
are as follows.
\begin{defn}
\label{def:=00005CC}$\C$ is considered \emph{well-fitting} if (1)
for any connection $\mathcal{R}_{q}\in\mathfrak{R}_{+}$ there is
one and only one coupling $S_{q}$ of $\mathcal{R}_{q}$ that satisfies
$\C$; and (2) if $\mathcal{R}_{q}$ consists of identically distributed
random variables, then the coupling that satisfies $\C$ is the identity
coupling. We denote such a coupling of $\mathcal{R}_{q}$ as $\C\left[\mathcal{R}_{q}\right]$.
\end{defn}
To give an example of a well-fitting statement $\C$: in CbD, if the
class $\mathfrak{\mathcal{\mathfrak{R}}}$ of all possible systems
is confined to the systems with dichotomous variables, the well-fitting
statement is $\C=$ ``for any two random variables $S_{q}^{c_{1}}$
and $S_{q}^{c_{2}}$ in the coupling of connection $\mathcal{R}_{q}$,
the probability of $S_{q}^{c_{1}}=S_{q}^{c_{2}}$ is maximal possible.''
Another example: if the class $\mathfrak{\mathcal{\mathfrak{R}}}$
of all possible systems is confined to the systems with real-valued
(or more generally, linearly ordered) variables, then a well-fitting
statement can be $\C=$ ``for any two random variables $S_{q}^{c_{1}}$
and $S_{q}^{c_{2}}$ in the coupling of connection $\mathcal{R}_{q}$,
$S_{q}^{c_{1}}$ and $S_{q}^{c_{2}}$ have the same quantile rank.''
In section \ref{subsec:existence-uniqueness} we will discuss the
possibility of dropping the first of the two defining properties of
a well-fitting statement $\C$.
\begin{defn}
\label{def:C-contextuality}Given a well-fitting $\C$, a system $\mathcal{R}$
is $\C$-noncontextual if it has a coupling $S$ such that, for any
connection $\mathcal{R}_{q}$ of the system, the connection coupling
$S_{q}$ coincides with $\C\left[\mathcal{R}_{q}\right]$. Otherwise
the system is $\C$-contextual.
\end{defn}

\section{\label{sec:Impossibility-theorem}Equivalence and impossibility theorems}

It follows from the last two definitions that, for a well-fitting
$\C$, a consistently connected system is $\C$-noncontextual if and
only if it is noncontextual in the traditional sense (i.e., in the
sense of Definition \ref{def:traditional}). In other words, any extension
of the notion of contextuality using a well-fitting $\C$ properly
reduces to the traditional notion when confined to consistently connected
systems. This is not, obviously, sufficient to consider the extension
of contextuality by means of $\C$ well-constructed. There may be
other desiderata for a well-constructed notion of contextuality, and
a specific choice of $\C$ may not satisfy them. The question we pose
now is as follows:
\begin{description}
\item [{Q{*}:}] is it possible to formulate a set of such desiderata/requirements
$\X$ for the notion of contextuality that, for some choice of $\C$,
(1) $\X$ is satisfied for any consistently connected system, but
(2) $\X$ is not satisfied for some inconsistently connected systems?
\end{description}
Note that we impose no constraints on what $\X$ may entail, except
for its being related to contextuality. It may, e.g., for some relation
$B$ between systems, have the form ``if system $\mathcal{R}_{1}$
is (non)contextual, then any system $\mathcal{R}_{2}$ related to
$\mathcal{R}_{1}$ by $B$ is (non)contextual'' \cite{DzhKuj2023Nonmeasurements}.

To answer the question \textbf{Q{*}} we need the following result.
\begin{thm}
\label{thm:Main}For any well-fitting\textup{ $\C$} and system $\mathcal{R}$,
there is a consistently connected system $\mathcal{R}^{\ddagger}$
that canonically models it (Definition \ref{def:modeling}), such
that $\mathcal{R}$ is\textup{ $\C$}-contextual (Definition \ref{def:C-contextuality})
if and only if $\mathcal{R}^{\ddagger}$ is contextual in the traditional
sense (Definition \ref{def:traditional}).
\end{thm}
\begin{proof}
Let $\mathcal{R}^{\ddagger}$ be a canonically modeling system for
$\mathcal{R}$, with
\[
\mathcal{R}^{\left(q,\cdot\right)}=\left\{ R_{\left(q,c\right)}^{\left(q,\cdot\right)}:\left(q,c\right)\prec^{\ddagger}\left(q,\cdot\right)\right\} \overset{d}{=}\mathsf{\C}\left[\left\{ R_{q}^{c}:q\prec c\right\} \right]=\C\left[\mathcal{R}_{q}\right].\tag{*}
\]
One can check that $\mathcal{R}^{\ddagger}$ is consistently connected:
every connection $\mathcal{\mathcal{R}}_{\left(q,c\right)}$ of $\mathcal{R}^{\ddagger}$
consists of precisely two variables, $R_{\left(q,c\right)}^{\left(\cdot,c\right)}$
and $R_{\left(q,c\right)}^{\left(q,\cdot\right)}$, where $R_{\left(q,c\right)}^{\left(\cdot,c\right)}\overset{d}{=}R_{\left(q,c\right)}^{\left(q,\cdot\right)}$.
Indeed, $R_{\left(q,c\right)}^{\left(\cdot,c\right)}\overset{d}{=}R_{q}^{c}$,
because $R^{\left(\cdot,c\right)}\overset{d}{=}R^{c}$ in any canonically
modeling system; and $R_{\left(q,c\right)}^{\left(q,\cdot\right)}\overset{d}{=}R_{q}^{c}$
because we know from ({*}) that $R_{\left(q,c\right)}^{\left(q,\cdot\right)}\overset{d}{=}S_{q}^{c}$,
where $S_{q}^{c}\in\C\left[\mathcal{R}_{q}\right]$.

The system $\mathcal{R}^{\ddagger}$ thus constructed is referred
to as a \emph{consistification} of $\mathcal{R}$. We can now define
the consistification $S^{\ddagger}$ of a coupling $S$ of a system
in precisely the same way as for the system itself, except that ({*})
is replaced with the straightforward 
\[
S^{\left(q,\cdot\right)}=S_{q},
\]
with the obvious correspondence between the different indexations
within the two random vectors. Clearly, $S^{\ddagger}$ is a coupling
of $\mathcal{R}^{\ddagger}$.

Assume now that $\mathcal{R}$ is noncontextual. This means that it
has a coupling $S$ such that (a) $S^{c}\overset{d}{=}R^{c}$ for
every $c\in C$, and (b) $S_{q}=\C\left[\mathcal{R}_{q}\right]$ for
every $q\in Q$. Then in the coupling $S^{\ddagger}$ of system $\mathcal{R}^{\ddagger}$
we have (a') $S^{\left(\cdot,c\right)}\overset{d}{=}R^{\left(\cdot,c\right)}$
for every $\left(\cdot,c\right)\in C^{\ddagger}$, and (b') $S^{\left(q,\cdot\right)}=\C\left[\mathcal{R}_{q}\right]$
for every $\left(q,\cdot\right)\in C^{\ddagger}$. Moreover, since
both $S_{\left(q,c\right)}^{\left(\cdot,c\right)}$ and $S_{\left(q,c\right)}^{\left(q,\cdot\right)}$
equal $S_{q}^{c}$, we have (c') $S_{\left(q,c\right)}^{\left(\cdot,c\right)}=S_{\left(q,c\right)}^{\left(q,\cdot\right)}$.
But (a')-(b')-(c') mean that $\mathcal{R}^{\ddagger}$ is noncontextual
in the traditional sense. The implication here is easily seen to be
reversible, and we conclude that $\mathcal{R}$ is noncontextual if
and only if so is $\mathcal{R}^{\ddagger}$.
\end{proof}
In our example (\ref{eq:QOE}), $\mathcal{B}$ is a consistification
of $\mathcal{A}$ if we specify the rule for the auxiliary bunches
as follows: $R_{\left(q,c\right)}^{\left(q,\cdot\right)}\overset{d}{=}R_{\left(q,c\right)}^{\left(\cdot,c\right)}$,
and the distribution of $R^{\left(q,\cdot\right)}$ is the same as
that of $\C\left[\mathcal{R}_{q}\right]$. If $\C$ is chosen as in
CbD, the consistification of the system $\mathcal{R}$ in (\ref{eq:example system})
is the system below (omitting for simplicity the parentheses and commas
in $R_{\left(q,c\right)}^{\left(\cdot,c\right)}$ and $R_{\left(q,c\right)}^{\left(q,\cdot\right)}$):
\begin{equation}
{\normalcolor \begin{array}{|c|c|c|c|c|c|c|c|c||c|}
\hline R_{11}^{\cdot1} & R_{21}^{\cdot1} &  &  &  &  &  &  &  & c=\cdot1\\
\hline  &  & R_{22}^{\cdot2} & R_{32}^{\cdot2} &  &  &  &  &  & c=\cdot2\\
\hline  &  &  &  & R_{13}^{\cdot3} & R_{33}^{\cdot3} &  &  &  & c=\cdot3\\
\hline  &  &  &  &  &  & R_{14}^{\cdot4} & R_{24}^{\cdot4} & R_{34}^{\cdot4} & c=\cdot4\\
\hline R_{11}^{1\cdot} &  &  &  & R_{13}^{1\cdot} &  & R_{14}^{1\cdot} &  &  & c=1\cdot\\
\hline  & R_{21}^{2\cdot} & R_{22}^{2\cdot} &  &  &  &  & R_{24}^{2\cdot} &  & c=2\cdot\\
\hline  &  &  & R_{32}^{3\cdot} &  & R_{33}^{3\cdot} &  &  & R_{34}^{3\cdot} & c=3\cdot\\
\hline\hline q=11 & q=21 & q=22 & q=32 & q=13 & q=33 & q=14 & q=24 & q=34 & \mathcal{R}^{\ddagger}
\\\hline \end{array}\:,}
\end{equation}
where all variables are assumed to be dichotomous, and in each of
the auxiliary bunches, the variables are pairwise equal with maximal
possible probability.

For the purposes of contextuality analysis, $\mathcal{R}^{\ddagger}$
can be viewed as a mere reformulation of $\mathcal{R}$, a different
form of the same substance. We express this fact by saying that $\mathcal{R}$
and $\mathcal{R}^{\ddagger}$ are \emph{contextually equivalent}.
(In Refs. \cite{chapter,DzhKuj2023Nonmeasurements} contextual equivalence
is defined more narrowly, requiring also the numerical coincidence
of certain measures of contextuality, such as contextual fraction
\cite{Abramskyetal.2017}. In this paper, however, the level of abstraction
is higher, and we only consider the notion of contextuality rather
than its quantifications.)

Consider now a theory of (generally, extended) contextuality $\T=\T\left(\C,\mathfrak{R}\right)$.
In accordance with Theorem $\ref{thm:Main}$, we can form the class
$\mathfrak{R}^{\ddagger}$ of the consistifications of the elements
of $\mathfrak{R}$ in a bijective correspondence with $\mathfrak{R}$.
By extension of the term, we can say that $\T$ and $\T'=\T\left(\C_{0},\mathfrak{\mathfrak{R}^{\ddagger}}\right)$
are \emph{contextually equivalent}. $\C_{0}$ here denotes the statement
``the connection $R_{q}$ has an identity coupling'' that underlies
the traditional notion of contextuality, because by definition, it
can be viewed as a special case of any well-fitting statement $\C$.
We have now everything we need to demonstrate our main conclusion.
Let there be a set of requirements $\X$ of the notion of contextuality
that are satisfied by all consistently connected systems (using the
traditional contextuality) and contravened by some inconsistently
connected ones, using some version of $\C$-contextuality. Let $\T$
include some of the inconsistently connected systems contravening
$\X$. Clearly then, requirements $\X$ contradict theory $\T$, but
they are satisfied by the contextually equivalent theory $\T'=\T\left(\C_{0},\mathfrak{\mathfrak{R}^{\ddagger}}\right)$.
Therefore $\X$ is not a set of substantive requirements. We can summarize
this as a formal theorem.
\begin{thm}
For any well-fitting \textup{$\C$}, there can be no set of substantive
requirements $\X$ of the notion of contextuality that are satisfied
by all consistently connected systems (using the traditional contextuality)
and contravened by some inconsistently connected ones, using $\C$-contextuality.
\end{thm}
Of course, a set of requirements $\X$ satisfied by $\T'$ but not
$\T$ can be readily formulated. The theorem says, however, that all
it can do is to lead one to prefer one of two equivalent representations
of contextuality, without affecting the substance of the notion.

Note also that in the theorem just formulated we assume no relationship
between the set of requirements $\X$ and the bijective correspondence
relating $\mathfrak{R}$ to $\mathfrak{\mathfrak{R}^{\ddagger}}$.
In particular, let $\X$ have the form ``if system $\mathcal{R}_{1}$
is contextual, then any system $\mathcal{R}_{2}$ related to $\mathcal{R}_{1}$
by relation $B$ is contextual.'' It is not necessary then, although
not excluded either, that $\mathcal{R}_{2}^{\ddagger}$ is also related
to $\mathcal{R}_{1}^{\ddagger}$ by relation $B$. All that is stated
in the theorem above is that if one wishes to use this $\X$ as a
substantive principle in testing competing theories, then the failure
of a theory to satisfy it cannot be selectively attributed to the
fact that its $\mathfrak{R}$ contains inconsistently connected systems.

\section{Miscellaneous remarks}

Here, we consider a few issues related to the main point of this paper.

\subsection{\label{subsec:Interpretation-of-contents}Interpretation of contents
and contexts}

Dealing with consistified systems $R^{\ddagger}$, one needs to get
used to a new interpretation of contents and contexts of the random
variables: as mentioned previously, in $R^{\ddagger}$, contents are
``contextualized,'' with $\left(q,c\right)$ in place of just $q$,
and the contexts are simply marginalized contents, $\left(\cdot,c\right)$
and $\left(q,\cdot\right)$. Consider as an example the EPR/Bohm experiment,
the most widely investigated paradigm in contextuality/nonlocality
research \cite{Bell1964,CHSH1969,Fine1982}. In the usual CbD notation,
the system representing it is
\begin{equation}
\begin{array}{|c|c|c|c||c|}
\hline R_{1}^{1} & R_{2}^{1} &  &  & c=1\\
\hline  & R_{2}^{2} & R_{3}^{2} &  & c=2\\
\hline  &  & R_{3}^{3} & R_{4}^{3} & c=3\\
\hline R_{1}^{4} &  &  & R_{4}^{4} & c=4\\
\hline\hline q=1 & q=2 & q=3 & q=4 & \mathcal{A}
\\\hline \end{array}\:,\label{eq:EPR1}
\end{equation}
where $q=1$ and $q=3$ denote two settings (axes) to be chosen between
by Alice, $q=2$ and $q=4$ are settings to choose between by Bob,
$c$ indicates the combination of their choices, and $R_{q}^{c}$
are dichotomous (spin-up/spin-down) variables. The consistified representation
of the same experiment is (again, omitting the parentheses and commas
in the indexation)
\begin{equation}
\begin{array}{|c|c|c|c|c|c|c|c||c|}
\hline R_{11}^{\cdot1} & R_{21}^{\cdot1} &  &  &  &  &  &  & c=\cdot1\\
\hline  &  & R_{22}^{\cdot2} & R_{32}^{\cdot2} &  &  &  &  & c=\cdot2\\
\hline  &  &  &  & R_{33}^{\cdot3} & R_{43}^{\cdot3} &  &  & c=\cdot3\\
\hline  &  &  &  &  &  & R_{44}^{\cdot4} & R_{14}^{\cdot4} & c=\cdot4\\
\hline R_{11}^{1\cdot} &  &  &  &  &  &  & R_{14}^{1\cdot} & c=1\cdot\\
\hline  & R_{21}^{2\cdot} & R_{22}^{2\cdot} &  &  &  &  &  & c=2\cdot\\
\hline  &  &  & R_{32}^{3\cdot} & R_{33}^{3\cdot} &  &  &  & c=3\cdot\\
\hline  &  &  &  &  & R_{43}^{4\cdot} & R_{44}^{4\cdot} &  & c=4\cdot\\
\hline\hline q=11 & q=21 & q=22 & q=32 & q=33 & q=43 & q=44 & q=14 & \mathcal{A}^{\ddagger}
\\\hline \end{array}\:.\label{eq:EPR2}
\end{equation}
The interpretation of, say, the content $q=\left(3,2\right)$ here
is as follows: it is the choice of axis 3 (that we know to be made
by Alice) when Bob's choice of his axis forms combination 2 with Alice's
choice (which we know to mean that Bob chooses axis 2). The interpretation
of context $c=\left(\cdot,2\right)$ is that it is simply the set
of contents whose second component is 2. Similarly, $c=\left(3,\cdot\right)$
is the set of contents whose first component is 3. The random variables
within context $c=\left(\cdot,2\right)$ are jointly distributed by
observation, whereas the random variables within context $c=\left(3,\cdot\right)$
are jointly distributed by computation (that, in turn, is uniquely
determined by the observations). If $\C$ is defined in accordance
with CbD, $\left(R_{21}^{2\cdot},R_{22}^{2\cdot}\right)$ is computed
so that $R_{21}^{2\cdot}\overset{d}{=}R_{21}^{\cdot1}$, $R_{22}^{2\cdot}\overset{d}{=}R_{22}^{\cdot2}$
(consistent connectedness), and the probability of $R_{21}^{2\cdot}=R_{22}^{2\cdot}$
is maximal possible. In particular, if $R_{21}^{\cdot1}\overset{d}{=}R_{22}^{\cdot2}$,
then $R_{21}^{2\cdot}=R_{22}^{2\cdot}$.

\subsection{Hidden variable models}

One possible argument against contextuality in inconsistently connected
systems is that it is not distinguishable from inconsistent connectedness
itself in the language of hidden variable models (HVMs). If, the argument
goes, a consistently connected system $\mathcal{R}$ in (\ref{eq:system})
is noncontextual, it has a coupling $S$ in which all random variables
can be presented as 
\begin{equation}
S_{q}^{c}=F\left(q,\Lambda\right),\label{eq:HVM0}
\end{equation}
where $\Lambda$ is a ``hidden'' random variable \cite{Dzh2019}.
If $\mathcal{R}$ is contextual, then all its couplings can only be
presented as 
\begin{equation}
S_{q}^{c}=F\left(q,c,\Lambda\right),\label{eq:HVM1}
\end{equation}
with ineliminable $c$. However, the latter representation is also
required for all inconsistently connected systems, irrespective of
whether they are $\C$-contextual or $\C$-noncontextual. We would
argue in response that this only means that on this general level
(merely showing the arguments of the functions) the language of HVMs
is too crude to capture the subtler properties of the couplings, such
as contextuality under inconsistent connectedness. However, even if
one takes this issue as a matter of concern, it is eliminated by the
consistification procedure. The system $\mathcal{R}^{\ddagger}$ corresponding
to $\mathcal{R}$ is noncontextual if and only if it has a coupling
$S^{\ddagger}$ such that, for all $\left(q,c\right)\in Q^{\ddagger},$
\begin{equation}
S_{\left(q,c\right)}^{\left(\cdot,c\right)}=G\left(\left(q,c\right),\Lambda\right)=S_{\left(q,c\right)}^{\left(q,\cdot\right)},\label{eq:HVM+}
\end{equation}
for some random variable $\Lambda$. If $\mathcal{R}^{\ddagger}$
is contextual, then in all its couplings, for some $\left(q,c\right)\in Q^{\ddagger},$$S_{\left(q,c\right)}^{\left(\cdot,c\right)}\not=S_{\left(q,c\right)}^{\left(q,\cdot\right)}$,
which means that their HVM representations can only be different functions,
\begin{equation}
S_{\left(q,c\right)}^{\left(\cdot,c\right)}=G_{1}\left(\left(q,c\right),\Lambda\right),S_{\left(q,c\right)}^{\left(q,\cdot\right)}=G_{2}\left(\left(q,c\right),\Lambda\right),\label{eq:HVM-}
\end{equation}
or, equivalently, the same function but with differently distributed
hidden variables,
\begin{equation}
S_{\left(q,c\right)}^{\left(\cdot,c\right)}=H\left(\left(q,c\right),\Lambda_{1}\right),S_{\left(q,c\right)}^{\left(q,\cdot\right)}=H\left(\left(q,c\right),\Lambda_{2}\right).
\end{equation}

It is instructive to apply this to the EPR/Bohm systems $\mathcal{A}$
and $\mathcal{A}^{\ddagger}$ in (\ref{eq:EPR1}) and (\ref{eq:EPR2}).
Here, contextuality is traditionally referred to as nonlocality, because
for the contextual system $\mathcal{A}$, all its couplings are represented
in the form of (\ref{eq:HVM1}): the ineliminable dependence on $c$
here is interpreted as the dependence of a measurement on a remote
setting. However, if one models the EPR/Bohm experiment by system
$\mathcal{A}^{\ddagger}$ instead, the HVM representations (\ref{eq:HVM+})
and (\ref{eq:HVM-}) both contain the contextualized content $\left(q,c\right)$
as an argument. Following the logic above, they should both be considered
nonlocal, even though one of them represents a noncontextual system
and is equivalent to (\ref{eq:HVM0}), while the other represents
a contextual system and is equivalent to (\ref{eq:HVM1}). It seems
to us, in agreement with other authors \cite{Khren2008}, that this
demonstration speaks against a naturalistic interpretation of the
HVMs in terms of physical dependences.

\subsection{\label{subsec:existence-uniqueness}The existence and uniqueness
constraint}

In the definition of $\C$-couplings, their reducibility to identity
couplings when applied to identically distributed variables is indispensable,
because without it the $\C$-contextuality will not be an extension
of traditional contextuality. How critical, however, is the second
constraint imposed on well-fitting $\C$, that the $\C\left[\mathcal{R}_{q}\right]$-coupling
always exists and is unique? What if one considers statements $\C$
for which $\C\left[\mathcal{R}_{q}\right]$ is a set that may be empty
or contain more than one coupling? This does complicate the matters
conceptually, because then, in the consistification procedure, the
$\left(q,\cdot\right)$-type bunches, those filled with the $\C\left[\mathcal{R}_{q}\right]$-couplings,
cannot be formed at all or cannot be formed uniquely. However, the
main point of this paper can still be made, with some qualifications.

We can agree that the consistification of an inconsistently connected
system $R$ is not a single system $R^{\ddagger}$ but a cluster of
systems $\left\{ R_{i}^{\ddagger}:i\in\mathbb{I}\right\} $, the elements
of which are obtaining by filling the $\left(q,\cdot\right)$-type
bunches in the consistification of $R$ by all possible couplings
of $R$'s connections. We can further agree that the cluster $\left\{ R_{i}^{\ddagger}:i\in\mathbb{I}\right\} $
is considered noncontextual if it contains a noncontextual system
$R_{i}^{\ddagger}$. In particular, if $\left\{ R_{i}^{\ddagger}:i\in\mathbb{I}\right\} $
is empty (which means that $\C\left[\mathcal{R}_{q}\right]$ does
not exist for at least one of the connections of $R$), the latter
definition is not satisfied, and the cluster should be considered
contextual. Once again, we have a theory dealing with consistently
connected systems only, except that the empirical or theoretical situations
they depict are represented by clusters of systems sharing a format
and the $\left(\cdot,c\right)$-bunches.

It might seem that dealing with an infinity of possible couplings
$\C\left[\mathcal{R}_{q}\right]$ or proving that $\C\left[\mathcal{R}_{q}\right]$
is empty is a significantly more difficult mathematical task than
when $\C$ is well-fitting. This is not the case, however, the complication
is not necessarily major. Mathematically, the problem of finding whether
a system $R$ is contextual consists in determining whether the variable
$S$ having the same format as $R$ can be assigned a probability
measure subject to certain constraints on its marginals. The constraints
are imposed by the distributions of the bunches $R^{c}$ (that $S^{c}$
have to match) and by the statement $\C$ that has to be satisfied
by the couplings $S_{q}$ of the connections $\mathcal{R}_{q}$. For
discrete random variables and finite sets $Q$ and $C$, this is a
linear programming task, provided the compliance with $\C$ can be
presented in terms of linear inequalities of the probabilities in
the distribution of $S_{q}$. For the consistification $R^{\ddagger}$
the problem is precisely the same, except that in place of connection
couplings one deals with $\left(q,\cdot\right)$-type bunches.

\paragraph*{Acknowledgements.}

This research was partially supported by the Foundational Questions
Institute grant FQXi-MGA-2201. We are grateful to Víctor H. Cervantes
and Damir D. Dzhafarov for critically commenting on the manuscript.

\end{document}